\documentclass[proceedings]{stacs}
\stacsheading{2010}{323-334}{Nancy, France}
\firstpageno{323}

\usepackage{amsmath,amssymb,amsthm}
\usepackage{cite}

\newcommand{\itemspacing}{\setlength{\itemsep}{0pt}}

\renewcommand{\geq}{\geqslant}
\renewcommand{\leq}{\leqslant}

\newcommand{\Bool}{\ensuremath{\{0,1\}}}

\newcommand{\imp}{\rightarrow}

\newcommand{\OR}{\ensuremath{\mathrm{OR}}}
\newcommand{\NAND}{\ensuremath{\mathrm{NAND}}}
\newcommand{\ORconj}{\ensuremath{\OR{}\text{-conj}}}
\newcommand{\NANDconj}{\ensuremath{\NAND{}\text{-conj}}}
\newcommand{\IMconj}{\ensuremath{\text{IM-conj}}}

\newcommand{\bwcomp}[1]{{\widetilde{#1}}}

\newcommand{\width}[1]{{\mathrm{wd}(#1)}}
\newcommand{\powerset}[1]{{\mathcal{P}({#1})}}

\newcommand{\pppleq}{\leq_{\mathrm{ppp}}}

\newcommand{\Ror}{R_{\mathrm{OR}}}
\newcommand{\Rnand}{R_{\mathrm{NAND}}}
\newcommand{\Rork}[1][k]{R_{\mathrm{OR}, {#1}}}
\newcommand{\Rnandk}[1][k]{R_{\mathrm{NAND}, {#1}}}
\newcommand{\Rimp}{R_{\imp}}
\newcommand{\Rpmi}{R_{\leftarrow}}
\newcommand{\Rneq}{R_{\neq}}
\newcommand{\Req}{R_{=}}
\newcommand{\Reqk}[1][k]{R_{{=},{#1}}}
\newcommand{\Rzero}{R_{\mathrm{zero}}}
\newcommand{\Rone}{R_{\mathrm{one}}}
\newcommand{\GammaPin}{\ensuremath{\Gamma_\mathrm{\!pin}}}

\newcommand{\tuple}[1]{\left\langle\,{#1}\,\right\rangle}

\newcommand{\abar}{\overline{a}}
\newcommand{\bbar}{\overline{b}}
\newcommand{\cbar}{\overline{c}}

\newcommand{\vbar}{\overline{v}}

\newcommand{\numBIS}{\ensuremath{\#\mathrm{BIS}}}
\newcommand{\numSAT}{\ensuremath{\#\mathrm{SAT}}}
\newcommand{\numwHIS}{\ensuremath{\#w\mathrm{\text{-}HIS}}}
\newcommand{\numwHISd}[1][d]{\ensuremath{\#w\mathrm{\text{-}HIS}_{#1}}}

\newcommand{\CSP}{\ensuremath{\mathrm{CSP}}}
\newcommand{\numCSP}{\ensuremath{\#\mathrm{CSP}}}
\newcommand{\numCSPd}[1][d]{\ensuremath{\numCSP_{#1}}}

\newcommand{\Ptime}{\ensuremath{\mathbf{P}}}
\newcommand{\numP}{\ensuremath{\#\mathbf{P}}}
\newcommand{\FPtime}{\ensuremath{\mathbf{FP}}}
\newcommand{\NPtime}{\ensuremath{\mathbf{NP}}}
\newcommand{\RPtime}{\ensuremath{\mathbf{RP}}}

\newcommand{\APredto}{\leq_{\mathrm{AP}}}
\newcommand{\APred}{\leq_{\mathrm{AP}}}
\newcommand{\APequiv}{\equiv_{\mathrm{AP}}}

\newcommand{\GFtwo}{\ensuremath{\mathrm{GF}_2}}

\title[Approximating Bounded-Degree Boolean \#CSP]%
      {The Complexity of Approximating\\
       Bounded-Degree Boolean \#CSP}
\thanks{Funded in part by the EPSRC grant ``The Complexity of
    Counting in Constraint Satisfaction Problems''.}

\author[Leeds]{M.\@ Dyer}{Martin Dyer}
\address[Leeds]{School of Computing, University of Leeds, Leeds,
                LS2~9JT, U.K.}
\email{{M.E.Dyer,D.M.Richerby}@leeds.ac.uk}

\author[Liv]{L.\@ A.\@ Goldberg}{Leslie Ann Goldberg}
\address[Liv]{Department of Computer Science, University of Liverpool,
              Liverpool, L69~3BX, U.K.}
\email{L.A.Goldberg@liverpool.ac.uk}

\author[Liv,Bris]{M.\@ Jalsenius}{Markus Jalsenius}
\address[Bris]{Current address: Department of Computer Science,
               University of Bristol, Merchant Venturers Building,
               Woodland Road, Bristol, BS8~1UB, U.K.}
\email{M.Jalsenius@bristol.ac.uk}

\author[Leeds]{D.\@ M.\@ Richerby}{David~Richerby}


\begin{document}

\begin{abstract}
    The degree of a CSP instance is the maximum number of times that a
    variable may appear in the scope of constraints. We consider the
    approximate counting problem for Boolean CSPs with bounded-degree
    instances, for constraint languages containing the two unary
    constant relations $\{0\}$ and $\{1\}$. When the maximum degree is
    at least $25$ we obtain a complete classification of the
    complexity of this problem.  It is exactly solvable in
    polynomial-time if every relation
    in the constraint language is affine. It is equivalent to the
    problem of approximately counting independent sets in bipartite
    graphs if every relation can be expressed as conjunctions of
    $\{0\}$, $\{1\}$ and binary implication.  Otherwise, there is no
    FPRAS unless $\NPtime = \RPtime$. For lower degree bounds,
    additional cases arise in which the complexity is related to the
    complexity of approximately counting independent sets in
    hypergraphs.
\end{abstract}

\maketitle


\section{Introduction}
\label{sec:Intro}

In the constraint satisfaction problem (CSP), we seek to assign values
from some domain to a set of variables, while satisfying given
constraints on the combinations of values that certain subsets of the
variables may take.  Constraint satisfaction problems are ubiquitous
in computer science, with close connections to graph theory, database
query evaluation, type inference, satisfiability, scheduling and
artificial intelligence \cite{KV2000:Conjunctive,
Kum1992:CSP-algorithms, Mon1974:Constraints}.  CSP can also be
reformulated in terms of homomorphisms between relational structures
\cite{FV1998:MMSNP} and conjunctive query containment in database
theory \cite{KV2000:Conjunctive}.  Weighted versions of CSP
appear in statistical physics, where they correspond to partition
functions of spin systems \cite{Wel1993:Complexity}.

We give formal definitions in Section~\ref{sec:Prelim} but, for now,
consider an undirected graph $G$ and the CSP where the domain is
$\{\mathrm{red}, \mathrm{green}, \mathrm{blue}\}$, the variables are
the vertices of $G$ and the constraints specify that, for every edge
$xy\in G$, $x$ and $y$ must be assigned different values.  Thus, in a
satisfying assignment, no two adjacent vertices are given the same
colour: the CSP is satisfiable if, and only if, the graph is
3-colourable.  As a second example, given a formula in 3-CNF, we can
write a system of constraints over the variables, with domain
$\{\mathrm{true}, \mathrm{false}\}$, that requires the assignment to
each clause to satisfy at least one literal.
Clearly, the resulting CSP is directly equivalent to the original
satisfiability problem.


\subsection{Decision CSP}
\label{sec:Decision}

In the \emph{uniform constraint satisfaction problem}, we are given
the set of constraints explicitly, as lists of allowable combinations
for given subsets of the variables; these lists can be considered as
relations over the domain.  Since it includes problems such as
3-\textsc{sat} and 3-\textsc{colourability}, uniform CSP is
\NPtime{}-complete.  However, uniform CSP also includes problems in
\Ptime{}, such as 2-\textsc{sat} and 2-\textsc{colourability}, raising
the natural question of what restrictions lead to tractable problems.
There are two natural ways to restrict CSP: we can restrict the form
of the instances and we can restrict the form of the constraints.

The most common restriction to CSP is to allow only certain fixed
relations in the constraints.  The list of allowed relations is
known as the \emph{constraint language} and we write $\CSP(\Gamma)$
for the so-called \emph{non-uniform} CSP in which each constraint
states that the values assigned to some tuple of variables must be a
tuple in a specified relation in $\Gamma$.

The classic example of this is Schaefer's dichotomy for Boolean
constraint languages $\Gamma$ (i.e., those with domain $\{0,1\}$;
often called ``generalized satisfiability'') \cite{Sch1978:Boolean}.
He showed that $\CSP(\Gamma)$ is in $\Ptime$ if $\Gamma$ is included
in one of six classes and is \NPtime{}-complete, otherwise.  More
recently, Bulatov has produced a corresponding dichotomy for the
three-element domain \cite{Bul2006:Ternary}.  These two results
restrict the size of the domain but allow relations of arbitrary arity
in the constraint language.  The converse restriction --- relations of
restricted arity, especially binary relations, over arbitrary finite
domains --- has also been studied in depth \cite{HN1990:H-color,
    HN2004:Homomorphisms}.

For all $\Gamma$ studied so far, $\CSP(\Gamma)$ has been either in
\Ptime{} or \NPtime{}-complete and Feder and Vardi have conjectured
that this holds for every constraint language \cite{FV1998:MMSNP}.
Ladner has shown that it is not the case that every problem in
\NPtime{} is either in \Ptime{} or \NPtime{}-complete since, if
$\Ptime{}\neq \NPtime{}$, there is an infinite, strict hierarchy
between the two \cite{Lad1975:Reducibility}.  However, there are
problems in \NPtime{}, such as graph Hamiltonicity and even
connectedness, that cannot be expressed as $\CSP(\Gamma)$ for any
finite $\Gamma\,$\footnote{This follows from results on the expressive
    power of existential monadic second-order logic
    \cite{FSV1995:Monadic}.} and Ladner's diagonalization does not
seem to be expressible in CSP \cite{FV1998:MMSNP}, so a dichotomy for
CSP appears possible.

Restricting the tree-width of instances has also been a fruitful
direction of research \cite{Fre1990:CSP-tw, KV2000:Games-CSP}.  In
contrast, little is known about restrictions on the degree of
instances, i.e., the maximum number of times that any variable may
appear.  Dalmau and Ford have shown that, for any fixed Boolean
constraint language $\Gamma$ containing the constant unary relations
$\Rzero=\{0\}$ and $\Rone=\{1\}$, the complexity of $\CSP(\Gamma)$ for
instances of degree at most three is exactly the same as the
complexity of $\CSP(\Gamma)$ with no degree restriction
\cite{DF2003:bdeg-gensat}. The case where variables may appear at most
twice has not yet been completely classified; it is known that
degree-2 $\CSP(\Gamma)$ is as hard as general $\CSP(\Gamma)$
whenever $\Gamma$ contains $\Rzero$ and $\Rone$ and some relation
that is not a $\Delta$-matroid \cite{Fed2001:Fanout}; the known
polynomial-time cases come from restrictions on the kinds of
$\Delta$-matroids that appear in $\Gamma$ \cite{DF2003:bdeg-gensat}.


\subsection{Counting CSP}

A generalization of classical CSP is to ask how many satisfying
solutions there are.  This is referred to as counting CSP, \numCSP{}.
Clearly, the decision problem is reducible to counting: if we can
efficiently count the solutions, we can efficiently determine whether
there is at least one.  The converse does not hold: for example, we
can determine in polynomial time whether a graph admits a perfect
matching but it is \numP{}-complete to count the perfect
matchings, even in a bipartite graph \cite{Val1979:Enumeration}.

\numP{} is the class of functions $f$ for which there is a
nondeterministic, polynomial-time Turing machine that has exactly
$f(x)$ accepting paths for input $x$ \cite{Val1979:Permanent}.  It is
easily seen that the counting version of any \NPtime{} decision
problem is in \numP{} and \numP{} can be considered the counting
``analogue'' of \NPtime{}.  Note, though that problems that are
\numP{}-complete under appropriate reductions are, under standard
complexity-theoretic assumptions, considerably harder than
\NPtime{}-complete problems: $\Ptime^{\numP}$ includes the whole of
the polynomial hierarchy \cite{Tod1989:PH}, whereas $\Ptime^{\NPtime}$
is generally thought not to.

Although no dichotomy is known for CSP, Bulatov has
recently shown that, for all $\Gamma\!$, $\numCSP(\Gamma)$ is either
computable in polynomial time or \numP{}-complete
\cite{Bul2008:Dichotomy}.  However, Bulatov's dichotomy sheds little
light on which constraint languages yield polynomial-time counting
CSPs and which do not.  The criterion of the dichotomy is based on
``defects'' in a certain infinite algebra built up from the
polymorphisms of $\Gamma$ and it is open whether the characterization
is even decidable.  It also seems not to apply to bounded-degree
\numCSP{}.

So, although there is a full dichotomy for $\numCSP(\Gamma)$, results
for restricted forms of constraint language are still of interest.
Creignou and Hermann have shown that only one of Schaefer's
polynomial-time cases for Boolean languages survives the transition to
counting: $\numCSP(\Gamma)\in\FPtime$ (i.e., has a polynomial time
algorithm) if $\Gamma$ is affine (i.e., each relation is the solution
set of a system of linear equations over \GFtwo{}) and is
\numP{}-complete, otherwise \cite{CH1996:Bool-numCSP}.  This result
has been extended to rational and even complex-weighted instances
\cite{WBool, CLXxxxx:Complex-numCSP} and, in the latter case, the
dichotomy is shown to hold for the restriction of the problem in which
instances have degree~$3$. This implies that the degree-3 problem
$\numCSPd[3](\Gamma)$ ($\numCSP(\Gamma)$ restricted to instances of
degree~3) is in \FPtime{} if $\Gamma$ is affine and is
\numP{}-complete, otherwise.


\subsection{Approximate counting}

Since $\numCSP(\Gamma)$ is very often \numP{}-complete, approximation
algorithms play an important role.  The key concept is that of a
\emph{fully polynomial randomized approximation scheme} (FPRAS).  This
is a randomized algorithm for computing some function $f(x)$, taking
as its input $x$ and a constant $\epsilon > 0$, and computing a value
$Y$ such that $e^{-\epsilon} \leq Y/f(x) \leq e^\epsilon$ with
probability at least $\tfrac{3}{4}$, in time polynomial in
both $|x|$ and ${\epsilon}^{-1}$. (See Section~\ref{sec:Prelim:Approx}.)

Dyer, Goldberg and Jerrum have classified the complexity of
approximately computing $\numCSP(\Gamma)$ for Boolean constraint
languages \cite{DGJ2007:Bool-approx}.  When all relations in $\Gamma$
are affine, $\numCSP(\Gamma)$ can be computed exactly in polynomial
time by the result of Creignou and Hermann discussed above
\cite{CH1996:Bool-numCSP}.  Otherwise, if every relation in $\Gamma$
can be defined by a conjunction of pins (i.e., assertions $v=0$ or
$v=1$) and Boolean implications, then $\numCSP(\Gamma)$ is as hard to
approximate as the problem \numBIS{} of counting independent sets in a
bipartite graph; otherwise, $\numCSP(\Gamma)$ is as hard to
approximate as the problem \numSAT{} of counting the satisfying truth
assignments of a Boolean formula.  Dyer, Goldberg, Greenhill and
Jerrum have shown that the latter problem is complete for \numP{}
under appropriate approximation-preserving reductions (see
Section~\ref{sec:Prelim:Approx}) and has no FPRAS unless $\NPtime =
\RPtime$ \cite{DGGJ2004:Approx}, which is thought to be unlikely.  The
complexity of \numBIS{} is currently open: there is no known FPRAS but
it is not known to be \numP{}-complete, either. \numBIS{} is known to
be complete for a logically-defined subclass of \numP{} with respect
to approximation-preserving reductions \cite{DGGJ2004:Approx}.


\subsection{Our result}

We consider the complexity of approximately solving
Boolean $\numCSP$ problems when instances have bounded
degree. Following Dalmau and Ford~\cite{DF2003:bdeg-gensat} and
Feder~\cite{Fed2001:Fanout} we consider the case in which
$\Rzero=\{0\}$ and $\Rone=\{1\}$ are available.  We proceed by showing
that any Boolean relation that is not definable as a conjunction of
ORs or NANDs can be used in low-degree instances to assert equalities
between variables.  Thus, we can side-step degree restrictions by
replacing high-degree variables with distinct variables asserted
to be equal.

Our main result, Corollary~\ref{cor:main}, is a trichotomy for the
case in which instances have maximum degree~$d$ for some $d\geq
25$. If every relation in~$\Gamma$ is affine, then $\numCSPd(\Gamma
\cup \{\Rzero,\Rone\})$ is solvable in polynomial time. Otherwise, if
every relation in $\Gamma$ can be defined as a conjunction of
$\Rzero$, $\Rone$ and binary implications, then $\numCSPd(\Gamma \cup
\{\Rzero,\Rone\})$ is equivalent in approximation complexity to
$\numBIS{}$. Otherwise, it has no FPRAS unless
$\NPtime=\RPtime$. Theorem~\ref{theorem:complexity} gives a partial
classification of the complexity when $d<25$. In the new cases that
arise here, the complexity is given in terms of the complexity of
counting independent sets in hypergraphs with bounded degree and
bounded hyper-edge size. The complexity of this problem is not fully
understood and we explain what is known about it in
Section~\ref{sec:Complexity}.


\section{Preliminaries}
\label{sec:Prelim}


\subsection{Basic notation}

We write $\abar$ for the tuple $\tuple{a_1, \dots, a_r}$, which we
often shorten to $\abar = a_1\dots a_r$.  We write $a^r$ for the
$r$-tuple $a\dots a$ and $\abar\bbar$ for the
tuple formed from the elements of $\abar$ followed by those of
$\bbar$.  The \emph{bit-wise complement} of a relation $R\subseteq
\Bool^r$ is the relation $\bwcomp{R} = \{\tuple{a_1\oplus 1, \dots,
a_r\oplus 1} \mid \abar\in R\}$, where $\oplus$ denotes addition
modulo~2.

We say that a relation $R$ is \emph{ppp-definable}\footnote{This
should not be confused with the concept of primitive positive
definability (pp-definability) which appears in algebraic
treatments of CSP and \numCSP{}, for example in the work of
Bulatov \cite{Bul2008:Dichotomy}.} in a relation $R'$ and write
$R\pppleq R'$ if $R$ can be obtained from $R'$ by some sequence of the
following operations:
\begin{itemize}
\itemspacing
\item permutation of columns (for notational convenience only);
\item pinning (taking sub-relations of the form $R_{i\mapsto c} =
    \{\abar\in R \mid a_i = c\}$ for some $i$ and some $c\in\Bool$); and
\item projection (``deleting the $i$th column'' to give the relation
    $\{a_1\dots a_{i-1}a_{i+1}\dots a_r \mid a_1\dots a_r\in R\}$).
\end{itemize}
It is easy to see that $\pppleq$ is reflexive and transitive and that,
if $R\pppleq R'\!$, then $R$ can be obtained from $R'$ by first
permuting the columns, then making some pins and then projecting.

We write $\Req = \{00, 11\}$, $\Rneq = \{01, 10\}$, $\Ror = \{01, 10,
11\}$, $\Rnand = \{00, 01, 10\}$, $\Rimp = \{00, 01, 11\}$ and $\Rpmi
= \{00, 10, 11\}$.  For $k\geq 2$, we write $\Reqk = \{0^k\!, 1^k\}$,
$\Rork = \Bool^k \setminus \{0^k\}$ and $\Rnandk = \Bool^k \setminus
\{1^k\}$ (i.e., $k$-ary equality, \OR{} and \NAND{}).


\subsection{Boolean constraint satisfaction problems}

A \emph{constraint language} is a set $\Gamma = \{R_1, \dots, R_m\}$
of named Boolean relations.  Given a set $V$ of variables, the set of
\emph{constraints} over $\Gamma$ is the set $\mathrm{Cons}(V,\Gamma)$
which contains $R(\vbar)$ for every relation $R\in\Gamma$ with arity
$r$ and every $\vbar\in V^r\!$. Note that $v=v'$ and $v\neq v'$ are
not constraints unless the appropriate relations are included in
$\Gamma\!$. The \emph{scope} of a constraint $R(\vbar)$ is the tuple
$\vbar$, which need not consist of distinct variables.

An \emph{instance} of the constraint satisfaction problem (CSP) over
$\Gamma$ is a set $V$ of variables and a set $C\subseteq
\mathrm{Cons}(V,\Gamma)$ of constraints.  An \emph{assignment} to a
set $V$ of variables is a function $\sigma\colon V\to \Bool$.  An
assignment to $V$ \emph{satisfies} an instance $(V, C)$ if
$\tuple{\sigma(v_1), \dots, \sigma(v_r)}\in R$ for every constraint
$R(v_1, \dots, v_r)$.  We write $Z(I)$ for the number of satisfying
assignments to a CSP instance $I$.  We study the counting CSP problem
$\numCSP(\Gamma)$, parameterized by $\Gamma\!$, in which we must
compute $Z(I)$ for an instance $I=(V, C)$ of CSP over $\Gamma$.

The \emph{degree} of an instance is the greatest number of times any
variable appears among its constraints.  Note that the variable $v$
appears twice in the constraint $R(v,v)$.  Our specific interest in
this paper is in classifying the complexity of bounded-degree counting
CSPs.  For a constraint language $\Gamma$ and a positive integer $d$,
define $\numCSPd(\Gamma)$ to be the restriction of $\numCSP(\Gamma)$
to instances of degree at most $d$.  Instances of degree~1 are
trivial.

\begin{theorem}
\label{thrm:degree-1}
    For any $\Gamma\!$, $\numCSPd[1](\Gamma)\in \FPtime$. \qed
\end{theorem}

When considering $\numCSPd$ for $d\geq 2$, we follow established
practice by allowing \emph{pinning} in the constraint language
\cite{DF2003:bdeg-gensat, Fed2001:Fanout}. We write $\Rzero=\{0\}$ and
$\Rone=\{1\}$ for the two singleton unary relations.  We refer to
constraints in $\Rzero$ and $\Rone$ as \emph{pins}.  To make notation
easier, we will sometimes write constraints using constants instead of
explicit pins.  That is, we will allow the constants 0 and~1 to appear
in the place of variables in the scopes of constraints.  Such
constraints can obviously be rewritten as a set of ``proper''
constraints, without increasing degree. We let $\GammaPin$ denote the
constraint language $\{\Rzero, \Rone\}$.


\subsection{Hypergraphs}

A \emph{hypergraph} $H=(V,E)$ is a set $V=V(H)$ of vertices and a set
$E = E(H)\subseteq \powerset{V}$ of non-empty \emph{hyper-edges}. The
\emph{degree} of a vertex $v\in V(H)$ is the number $d(v) = |\{e\in
E(H)\mid v\in e\}|$ and the degree of a hypergraph is the maximum
degree of its vertices. If $w = \max \{|e| \mid e\in E(H)\}$, we say
that $H$ has \emph{width} $w$.  An \emph{independent set} in a
hypergraph $H$ is a set $S\subseteq V(H)$ such that $e\nsubseteq S$
for every $e\in E(H)$.  Note that an independent set may contain more
than one vertex from any hyper-edge of size at least three.

We write \numwHIS{} for the problem of counting the independent sets
in a width-$w$ hypergraph $H$, and \numwHISd{} for the restriction of
\numwHIS{} to inputs of degree at most $d$.


\subsection{Approximation complexity}
\label{sec:Prelim:Approx}

A \emph{randomized approximation scheme} (RAS) for a function $f\colon\Sigma^*\rightarrow\mathbb{N}$ is a probabilistic Turing machine that takes as input a pair $(x,\epsilon)\in \Sigma^*\times (0,1)$, and produces, on an output tape, an integer random variable~$Y$ with $\Pr(e^{-\epsilon} \leq Y/f(x) \leq e^\epsilon)\geq \frac{3}{4}$.%
       \footnote{The choice of the value $\frac{3}{4}$ is inconsequential: the same class of problems has an FPRAS if we choose any probability $p$ with $\frac{1}{2}<p<1$ \cite{JVV1986:Randgen}.}
A \emph{fully polynomial randomized approximation scheme (FPRAS)} is a
RAS that runs in time $\mathrm{poly}(|x|,\epsilon^{-1})$.

To compare the complexity of approximate counting problems, we use the AP-reductions of \cite{DGGJ2004:Approx}. Suppose $f$ and $g$ are two functions from some input domain $\Sigma^*$ to the natural numbers and we wish to compare the complexity of approximately computing~$f$ to that of approximately computing~$g$. An \emph{approximation-preserving} reduction from~$f$ to~$g$ is a probabilistic oracle Turing machine $M$ that takes as input a pair $(x,\epsilon)\in \Sigma^*\times (0,1)$, and satisfies the following three conditions: (i) every oracle call made by $M$ is of the form $(w,\delta)$ where $w\in \Sigma^*$ is an instance of~$g$, and $0<\delta<1$ is an error bound satisfying $\delta^{-1} \leq \mathrm{poly}(|x|,\epsilon^{-1})$; (ii) $M$ is a randomized approximation scheme for $f$ whenever the oracle is a randomized approximation scheme for $g$; and (iii) the run-time of $M$ is polynomial in $|x|$ and $\epsilon^{-1}$.
 
If there is an approximation-preserving reduction from $f$ to $g$, we write $f\APred g$ and say that $f$ is \emph{AP-reducible} to $g$. If $g$ has an FPRAS, then so does $f$. If $f\APred g$ and $g\APred f$, then we say that $f$ and $g$ are \emph{AP-interreducible} and write $f\APequiv g$.


\section{Classes of relations}
\label{sec:Relations}

A relation $R\subseteq \Bool^r$ is \emph{affine} if it is the set of
solutions to some system of linear equations over \GFtwo{}.  That is,
there is a set $\Sigma$ of equations in variables $x_1, \dots, x_r$,
each of the form $x_{i_1} \oplus \dots \oplus x_{i_n} = c$, where
$\oplus$ denotes addition modulo~2 and $c\in \Bool$, such that $\abar\in
R$ if, and only if, the assignment $x_1\mapsto a_1, \dots, x_r\mapsto
a_r$ satisfies every equation in $\Sigma$.  Note that the empty and complete
relations are affine.

We define \IMconj{} to be the class of relations defined by a
conjunction of pins and (binary) implications.  This class is called
$\text{IM}_2$ in \cite{DGJ2007:Bool-approx}.

\begin{lemma}
\label{lemma:IMconj-implies}
    If $R\in\IMconj$ is not affine, then $\Rimp\pppleq R$.\qed
\end{lemma}

Let \ORconj{} be the set of Boolean relations that are defined by a
conjunction of pins and \OR{}s of any arity and \NANDconj{} the set of
Boolean relations definable by conjunctions of pins and \NAND{}s
(i.e., negated conjunctions) of any arity.  We say that
one of the defining formulae of these relations is \emph{normalized}
if no pinned variable appears in any \OR{} or \NAND{}, the arguments
of each individual \OR{} and \NAND{} are distinct, every \OR{} or
\NAND{} has at least two arguments and no \OR{} or \NAND{}'s arguments
are a subset of any other's.

\begin{lemma}
\label{lemma:conj-norm}
    Every \ORconj{} (respectively, \NANDconj{}) relation is defined by
    a unique normalized formula.\qed
\end{lemma}

Given the uniqueness of defining normalized formulae, we define the
\emph{width} of an \ORconj{} or \NANDconj{} relation $R$ to be
$\width{R}$, the greatest number of arguments to any of the \OR{}s or
\NAND{}s in the normalized formula that defines it.  Note that, from
the definition of normalized formulae, there are no relations of
width~1.

\begin{lemma}
\label{lemma:ORconj-OR}
    If $R\in\ORconj$ has width $w$, then $\Rork[2], \dots, \Rork[w]
    \pppleq R$. Similarly, if $R\in\NANDconj$ has width $w$, then
    $\Rnandk[2], \dots, \Rnandk[w] \pppleq R$.\qed
\end{lemma}

Given tuples $\abar, \bbar\in \Bool^r\!$, we write $\abar\leq \bbar$
if $a_i\leq b_i$ for all $i\in [1,r]$.  If $\abar\leq \bbar$ and
$\abar \neq \bbar$, we write $\abar < \bbar$.  We say that a relation
$R\subseteq \Bool^r$ is \emph{monotone} if, whenever $\abar\in R$ and
$\abar\leq \bbar$, then $\bbar\in R$.  We say that $R$ is
\emph{antitone} if, whenever $\abar\in R$ and $\bbar\leq \abar$, then
$\bbar\in R$.  Clearly, $R$ is monotone if, and only if, $\bwcomp{R}$
is antitone.  Call a relation \emph{pseudo-monotone} (respectively,
\emph{pseudo-antitone}) if its restriction to non-constant columns is
monotone (respectively, antitone).  The following is a consequence of
results in \cite[Chapter~7.1.1]{KnuXXXX:TAOCPv4A}.

\begin{proposition}
\label{prop:OR-monotone}
    A relation $R\subseteq \Bool^r$ is in \ORconj{} (respectively,
    \NANDconj) if, and only if, it is pseudo-monotone (respectively,
    pseudo-antitone).\qed
\end{proposition}


\section{Simulating equality}
\label{sec:SimEq}

An important ingredient in bounded-degree dichotomy theorems
\cite{CLXxxxx:Complex-numCSP} is expressing equality using constraints
from a language that does not necessarily include the equality
relation.

A constraint language $\Gamma$ is said to \emph{simulate} the $k$-ary
equality relation $\Reqk$ if, for some $\ell \geq k$, there is a
$(\Gamma \cup \GammaPin)$-CSP instance $I$ with variables $x_1, \dots,
x_\ell$ that has exactly $m\geq 1$ satisfying assignments
$\sigma$ with $\sigma(x_1) = \dots = \sigma(x_k) = 0$, exactly $m$
with $\sigma(x_1) = \dots = \sigma(x_k) = 1$ and no other satisfying
assignments.  If, further, the degree of $I$ is $d$ and the degree of
each variable $x_1, \dots, x_k$ is at most $d-1$, we say that $\Gamma$
\emph{$d$-simulates} $\Reqk$.  We say that $\Gamma$
\emph{$d$-simulates equality} if it $d$-simulates $\Reqk$ for all
$k\geq 2$.

The point is that, if $\Gamma$ $d$-simulates equality, we can express
the constraint $y_1 = \dots = y_r$ in $\Gamma \cup \GammaPin$ and then
use each $y_i$ in one further constraint, while still having an
instance of degree $d$.  The variables $x_{k+1}, \dots, x_\ell$ in the
definition function as auxiliary variables and are not used in any
other constraint.  Simulating equality makes degree bounds moot.

\begin{proposition}
\label{prop:bound-unbound}
    If $\Gamma$ $d$-simulates equality, then $\numCSP(\Gamma) \APredto
    \numCSPd(\Gamma \cup \GammaPin)$.\qed
\end{proposition}

We now investigate which relations simulate equality.

\begin{lemma}
\label{lemma:proj-eq}
    $R\in\Bool^r$ 3-simulates equality if $\Req\pppleq R$,
    $\Rneq\pppleq R$ or $\Rimp\pppleq R$.
\end{lemma}
\begin{proof}
    For each $k\geq 2$, we show how to 3-simulate $\Reqk$.  We may
    assume without loss of generality that the ppp-definition of
    $\Req$, $\Rneq$ or $\Rimp$ from $R$ involves applying the identity
    permutation to the columns, pinning columns 3 to $3+p-1$ inclusive
    to zero, pinning columns $3+p$ to $3+p+q-1$ inclusive to one (that
    is, pinning $p\geq 0$ columns to zero and $q\geq 0$ to one)
    and then projecting away all but the first two columns.

    Suppose first that $\Req\pppleq R$ or $\Rimp\pppleq R$.  $R$
    must contain $\alpha\geq 1$ tuples that begin $000^p1^q$,
    $\beta\geq 0$ that begin $010^p1^q$ and $\gamma\geq 1$ that begin
    $110^p1^q$, with $\beta=0$ unless we are ppp-defining $\Rimp$.
    We consider, first, the case where $\alpha=\gamma$, and show that
    we can 3-simulate $\Reqk$, expressing the constraint $\Reqk(x_1,
    \dots, x_k)$ with the constraints
    \begin{equation*}
        R(x_1 x_2 0^p 1^q *), \ R(x_2 x_3 0^p 1^q *), \dots,
	    \ R(x_{k-1} x_k 0^p 1^q *), \ R(x_k x_1 0^p 1^q *) \,,
    \end{equation*}
    where $*$ denotes a fresh $(r-2-p-q)$-tuple of variables in each
    constraint.  These constraints are equivalent to $x_1 = \dots =
    x_k = x_1$ or to $x_1\imp \dots \imp x_k \imp x_1$ so constrain
    the variables $x_1, \dots, x_k$ to have the same value, as
    required.  Every variable appears at most twice and there are
    $\alpha^k$ solutions to these constraints that put
    $x_1=\dots=x_k=0$, $\gamma^k=\alpha^k$ solutions with $x_1=\dots=x_k=1$ and no
    other solutions.  Hence, $R$ 3-simulates $\Reqk$, as required.

    We now show, by induction on $r$, that we can 3-simulate $\Reqk$
    even in the case that $\alpha\neq\gamma$.  For the base case,
    $r=2$, we have $\alpha=\gamma=1$ and we are done.  For the
    inductive step, let $r>2$ and assume, w.l.o.g.\@ that
    $\alpha>\gamma$ ($\alpha<\gamma$ is symmetric).  In particular, we
    have $\alpha\geq 2$, so there are distinct tuples $000^p1^q\abar$,
    and $000^p1^q\bbar$ and $110^p1^q\cbar$ in $R$.  Choose $j$ such
    that $a_j\neq b_j$.  Pinning the $(2+p+q+j)$th column of $R$ to
    $c_j$ and projecting out the resulting constant column gives a
    relation $R'$ of arity $r-1$ containing at least one tuple
    beginning $000^p1^q$ and at least one beginning $110^p1^q$: by the
    inductive hypothesis, $R'$ 3-simulates $\Reqk$.

    Finally, we consider the case that $\Rneq\pppleq R$.  $R$ contains
    $\alpha\geq 1$ tuples beginning $010^p1^q$ and $\beta\geq 1$
    beginning $100^p1^q$.  We express the constraint $\Reqk(x_1,
    \dots, x_k)$ by introducing fresh variables $y_1, \dots, y_k$ and
    using the constraints
    \begin{equation*}
        \begin{array}{ccccc}
            R(x_1 y_1 0^p1^q*), & R(x_2 y_2 0^p1^q*), & \ldots,
                & R(x_{k-1} y_{k-1} 0^p1^q*), & R(x_k y_k 0^p1^q*), \\
            R(y_1 x_2 0^p1^q*), & R(y_2 x_3 0^p1^q*), & \ldots,
                & R(y_{k-1} x_k 0^p1^q*), & R(y_k x_1 0^p1^q*)\,.
        \end{array}
    \end{equation*}
    There are $\alpha^k \beta^k$ solutions when $x_1 = \dots = x_k =
    0$ (and $y_1 = \dots = y_k = 1$) and $\beta^k \alpha^k$ solutions
    when the $x$s are~1 and the $y$s are~0.  There are no other
    solutions and no variable is used more than twice.
\end{proof}

For $c\in \Bool$, an $r$-ary relation is \emph{$c$-valid} if it
contains the tuple $c^r\!$.

\begin{lemma}
\label{lemma:valid-eq}
    Let $r\geq 2$ and let $R\subseteq \Bool^r$ be 0- and 1-valid but
    not complete. Then $R$ 3-simulates equality.\qed
\end{lemma}

In the following lemma, we do not require $R$ and $R'$ to be distinct.
The technique is to assert $x_1=\dots=x_k$ by simulating the formula
$\OR(x_1,y_1) \wedge \NAND(y_1,x_2) \wedge \OR(x_2,y_2) \wedge
\NAND(y_2,x_3) \wedge \cdots \wedge \OR(x_k,y_k) \wedge \NAND(y_k,
x_1)$.

\begin{lemma}
\label{lemma:OR-NAND-eq}
    If $\Ror\pppleq R$ and $\Rnand\pppleq R'\!$, then $\{R,R'\}$
    3-simulates equality.\qed
\end{lemma}


\section{Classifying relations}
\label{sec:Trichotomy}

We are now ready to prove that every Boolean relation $R$ is in
\ORconj{}, in \NANDconj{} or 3-simulates equality. If
$R_0$ and $R_1$ are $r$-ary, let $R_0+R_1 =
\{0\abar\mid \abar\in R_0\} \cup \{1\abar\mid \abar\in R_1\}$.

\begin{lemma}
\label{lemma:OR-OR}
    Let $R_0, R_1\in\ORconj$ and let $R=R_0+R_1$. Then $R\in\ORconj$,
    $R\in\NANDconj$ or $R$ 3-simulates equality.
\end{lemma}
\begin{proof}
    Let $R_0$ and $R_1$ have arity $r$.  We may assume that $R$ has no
    constant columns.  If it does, let $R'$ be the relation that
    results from projecting them away. $R' = R'_0 + R'_1$, where both
    $R'_0$ and $R'_1$ are $\ORconj$ relations. By the remainder of the
    proof, $R'\in\ORconj$, $R'\in \NANDconj$ or $R'$ 3-simulates
    equality.  Re-instating the constant columns does not alter this.
    For $R$ without constant columns, there are two cases.

    \smallskip
    \noindent
    \emph{{\bf Case 1.} $R_0\subseteq R_1$.}  Suppose $R_i$ is defined
        by the normalized \ORconj{} formula $\phi_i$ in variables
        $x_2, \dots, x_{r+1}$.  Then $R$ is defined by the formula
        \begin{equation}
            \phi_0 \vee (x_1=1 \wedge \phi_1)
            \equiv (\phi_0 \vee x_1=1)
                \wedge (\phi_0 \vee \phi_1)
            \equiv (\phi_0 \vee x_1=1) \wedge \phi_1\,,
                \label{eq:OR-conj-formula}
        \end{equation}
        where the second equivalence is because $\phi_0$ implies
        $\phi_1$, because $R_0\subseteq R_1$. $R_1$ has no constant
        column, since such a column would have to be constant with the
        same value in $R_0$, contradicting our assumption that $R$ has
        no constant columns.  There are two cases.

    \smallskip
    \noindent
    \emph{{\bf Case 1.1.} $R_0$ has no constant columns.}  $x_1=1$ is
        equivalent to $\OR(x_1)$ and $\phi_0$ contains no pins, so we
        can rewrite $\phi_0 \vee x_1=1$ in CNF.  Therefore,
        (\ref{eq:OR-conj-formula}) is \ORconj{}.

    \smallskip
    \noindent
    \emph{{\bf Case 1.2.} $R_0$ has a constant column.}  Suppose first
        that the $k$th column of $R_0$ is constant-zero.  $R_1$ has no
        constant columns, so the projection of $R$ onto its first and
        $(k+1)$st columns gives the relation $\Rpmi$, and $R$
        3-simulates equality by Lemma~\ref{lemma:proj-eq}. Otherwise,
        all constant columns of $R_0$ contain ones. Then $\phi_0$ is
        in CNF, since every pin $x_i=1$ in $\phi_0$ can be written
        $\OR(x_i)$. Thus, we can write $\phi_0 \vee x_1=1$ in
        CNF, so (\ref{eq:OR-conj-formula}) defines
        an \ORconj{} relation.

    \smallskip
    \noindent
    \emph{{\bf Case 2.} $R_0\nsubseteq R_1$.}
        We will show that $R$
        3-simulates equality or is in \NANDconj{}. We consider two
        cases (recall that no relation has width~1).

    \smallskip
    \noindent
    \emph{{\bf Case 2.1.} At least one of $R_0$ and $R_1$ has positive width.}
        There are two sub-cases.

    \smallskip
    \noindent
    \emph{{\bf Case 2.1.1.} $R_1$ has a constant column.}
        Suppose the $k$th column of $R_1$ is constant. If the $k$th column of $R_0$ is also constant, then the projection of $R$ to its first and $(k+1)$st columns is either equality or disequality (since the corresponding column of $R$ is not constant) so $R$ 3-simulates equality by Lemma~\ref{lemma:proj-eq}.  Otherwise, if the projection of $R$ to the first and $(k+1)$st columns is $\Rimp$, then $R$ 3-simulates equality by Lemma~\ref{lemma:proj-eq}.  Otherwise, that projection must be $\Rnand$.  By Lemma~\ref{lemma:ORconj-OR} and the assumption of Case~2.1, $\Ror$ is ppp-definable in at least one of $R_0$ and $R_1$ so $R$ 3-simulates equality by Lemma~\ref{lemma:OR-NAND-eq}.

    \smallskip
    \noindent
    \emph{{\bf Case 2.1.2.} $R_1$ has no constant columns.}
        By Proposition~\ref{prop:OR-monotone}, $R_1$ is monotone. Let $\abar\in R_0\setminus R_1$: by applying the same permutation to the columns of $R_0$ and $R_1$, we may assume that $\abar = 0^\ell 1^{r-\ell}$.  We must have $\ell\geq 1$ as every non-empty $r$-ary monotone relation contains the tuple $1^r\!$.  Let $\bbar\in R_1$ be a tuple such that $a_i=b_i$ for a maximal initial segment of $[1,r]$.  By monotonicity of $R_1$, we may assume that $\bbar = 0^k 1^{r-k}$.  Further, we must have $k<\ell$, since, otherwise, we would have $\bbar<\abar$, contradicting our choice of $\abar\notin R_1$.

        Now, consider the relation
            $R' = \{a_0 a_1\dots a_{\ell-k}\mid
                            a_00^ka_1\dots a_{\ell-k}1^{r-\ell} \in R\}$,
which is the result of pinning columns 2 to $(k+1)$ of $R$ to zero and columns $(r-\ell+1)$ to $(r+1)$ to one and discarding the resulting constant columns.  $R'$ contains $0^{\ell-k+1}$ and $1^{\ell-k+1}$ but is not complete, since $10^{\ell-k}\notin R'\!$.  By Lemma~\ref{lemma:valid-eq}, $R'$ and, hence, $R$ 3-simulates equality.

    \smallskip
    \noindent
    \emph{{\bf Case 2.2.} Both $R_0$ and $R_1$ have width zero,} i.e.,
        are complete relations, possibly padded with constant
        columns. For $i \in [1,r]$, let $R'_i$ be the relation
        obtained from $R$ by projecting onto its first and $(i+1)$st
        columns. Since $R$ has no constant columns, $R'_i$ is either
        complete, $\Req$, $\Rneq$, $\Ror$, $\Rnand$, $\Rimp$ or
        $\Rpmi$. If there is a $k$ such that $R'_k$ is $\Req$,
        $\Rneq$, $\Rimp$ or $\Rpmi$, then $\Req$, $\Rneq$ or $\Rimp$ is
        ppp-definable in $R$ and hence $R$ 3-simulates equality by
        Lemma~\ref{lemma:proj-eq}. If there are $k_1$ and $k_2$ such
        that $R'_{k_1} = \Ror$ and $R'_{k_2} = \Rnand$, then $R$
        3-simulates equality by Lemma~\ref{lemma:OR-NAND-eq}. It
        remains to consider the following two cases.

    \smallskip
    \noindent
    \emph{{\bf Case 2.2.1.} Each $R'_i$ is either $\Ror$ or complete.}
        $R_1$ must be complete, which contradicts the assumption that $R_0\not\subseteq R_1$.

    \smallskip
    \noindent
    \emph{{\bf Case 2.2.1.} Each $R'_i$ is either $\Rnand$ or complete.}
        $R_0$ must be complete. Let $I = \{i\mid R'_i=\Rnand\}$. Then
        $R = \bigwedge_{i\in I} \NAND(x_1,x_{i+1})$, so $R\in \NANDconj$.
\end{proof}

Using the duality between \ORconj{} and \NANDconj{} relations, we can
prove the corresponding result for $R_0, R_1\in \NANDconj$.  The proof
of the classification is completed by a simple induction on the arity
of $R$.  Decomposing $R$ as $R_0+R_1$ and assuming inductively that
$R_0$ and $R_1$ are of one of the stated types, we use the previous
results in this section and Lemma~\ref{lemma:OR-NAND-eq} to show that
$R$ is.

\begin{theorem}
\label{thrm:trichotomy}
    Every Boolean relation is \ORconj{} or \NANDconj{} or
    3-simulates equality.\qed
\end{theorem}


\section{Complexity}
\label{sec:Complexity}

The complexity of approximating $\numCSP(\Gamma)$ where the degree of
instances is unbounded is given by Dyer, Goldberg and Jerrum
\cite[Theorem~3]{DGJ2007:Bool-approx}.

\begin{theorem}
\label{thrm:unbounded}
    Let $\Gamma$ be a Boolean constraint language.
    \begin{itemize}
    \itemspacing
    \item If every $R\in\Gamma$ is affine, then $\numCSP(\Gamma)\in
        \FPtime$.
    \item Otherwise, if $\Gamma\subseteq\IMconj$, then $\numCSP(\Gamma)
        \APequiv \numBIS$.
    \item Otherwise, $\numCSP(\Gamma) \APequiv \numSAT$.
    \end{itemize}
\end{theorem}

Working towards our classification of the approximation complexity of
$\numCSP(\Gamma)$, we first deal with subcases.  The \IMconj{} case
and \ORconj{}/\NANDconj{} cases are based on links between those
classes of relations and the problems of counting independent sets in
bipartite and general graphs, respectively\cite{DGJ2007:Bool-approx,
DGGJ2004:Approx}, the latter extended to hypergraphs.

\begin{proposition}
\label{prop:im-bis}
    If $\Gamma \subseteq \IMconj$ contains at least one non-affine
    relation, then $\numCSPd(\Gamma\cup \GammaPin) \APequiv \numBIS$
    for all $d\geq 3$. \qed
\end{proposition}

\begin{proposition}
\label{prop:HIS-to-ORconj}
    Let $R$ be an \ORconj{} or \NANDconj{} relation of width~$w$. Then,
    for $d\geq 2$, $\numwHISd{} \APredto \numCSPd(\{R\}\cup
    \GammaPin)$.\qed
\end{proposition}

\begin{proposition}
\label{prop:ORconj-to-HIS}
    Let $R$ be an \ORconj{} or \NANDconj{} relation of width~$w$. Then, for $d\geq 2$, $\numCSPd(\{R\}\cup \GammaPin) \APredto \numwHISd[kd]$, where $k$ is the greatest number of times that any variable appears in the normalized formula defining $R$. \qed
\end{proposition}

We now give the complexity of approximating $\numCSPd(\Gamma\cup \GammaPin)$ for $d\geq 3$.

\begin{theorem}
\label{theorem:complexity}
    Let $\Gamma$ be a Boolean constraint language and let $d\geq 3$.
    \begin{itemize}
    \itemspacing
    \item If every $R\in\Gamma$ is affine, then $\numCSPd(\Gamma\cup \GammaPin) \in
        \FPtime$.
    \item Otherwise, if $\Gamma\subseteq \IMconj$, then
        $\numCSPd(\Gamma\cup \GammaPin) \APequiv \numBIS$.
    \item Otherwise, if $\Gamma\subseteq \ORconj$ or $\Gamma\subseteq
        \NANDconj$, then let $w$ be the greatest width of any relation in $\Gamma$ and let $k$ be the greatest number of times that any variable appears in the normalized formulae defining the relations of $\Gamma$. Then $\numwHISd \APredto \numCSPd(\Gamma\cup \GammaPin) \APredto \numwHISd[kd]$.
    \item Otherwise, $\numCSPd(\Gamma\cup \GammaPin) \APequiv \numSAT$.
    \end{itemize}
\end{theorem}
\begin{proof}
    The affine case is immediate from Theorem~\ref{thrm:unbounded}.
    ($\Gamma\cup \GammaPin$ is affine if, and only if, $\Gamma$ is.)
    Otherwise, if $\Gamma\subseteq \IMconj$ and some $R\in \Gamma$ is
    not affine, then $\numCSPd(\Gamma\cup \GammaPin) \APequiv \numBIS$
    by Proposition~\ref{prop:im-bis}.  Otherwise, if $\Gamma\subseteq
    \ORconj$ or $\Gamma\subseteq \NANDconj$, then $\numwHISd \APredto
    \numCSPd(\Gamma\cup \GammaPin) \APredto \numwHISd[kd]$ by
    Propositions \ref{prop:HIS-to-ORconj}
    and~\ref{prop:ORconj-to-HIS}.

    Finally, suppose that $\Gamma$ is not affine, $\Gamma\nsubseteq
    \IMconj$, $\Gamma\nsubseteq \ORconj$ and $\Gamma\nsubseteq
    \NANDconj$. Since $(\Gamma\cup \GammaPin)$ is neither affine or a
    subset of \IMconj{}, we have $\numCSP(\Gamma\cup
    \GammaPin)\APequiv \numSAT$ by Theorem~\ref{thrm:unbounded} so, if
    we can show that $\Gamma$ $d$-simulates equality, then
    $\numCSPd(\Gamma\cup \GammaPin) \APequiv \numCSP(\Gamma\cup
    \GammaPin)$ by Proposition~\ref{prop:bound-unbound} and we are
    done.  If $\Gamma$ contains a $R$ relation that is neither
    \ORconj{} nor \NANDconj{}, then $R$ 3-simulates equality by
    Theorem~\ref{thrm:trichotomy}.  Otherwise, $\Gamma$ must contain
    distinct relations $R_1\in\ORconj$ and $R_2\in\NANDconj$ that are
    non-affine so have width at least two.  So $\Gamma$ 3-simulates
    equality by Lemma~\ref{lemma:OR-NAND-eq}.
\end{proof}

Unless $\NPtime=\RPtime$, there is no FPRAS for counting independent
sets in graphs of maximum degree at least~25 \cite{DFJ2002:IS-sparse},
and, therefore, no FPRAS for $\numwHISd[d]$ with $r\geq 2$ and $d\geq
25$. Further, since $\numSAT$ is complete for $\numP$ under
AP-reductions \cite{DGGJ2004:Approx}, $\numSAT$ cannot have an FPRAS
unless $\NPtime=\RPtime$. From Theorem~\ref{theorem:complexity} above
we have the following corollary.

\begin{corollary}
\label{cor:main}
    Let $\Gamma$ be a Boolean constraint language and let $d\geq 25$.
    \begin{itemize}
    \itemspacing
    \item If every $R\in\Gamma$ is affine, then $\numCSPd(\Gamma\cup \GammaPin) \in \FPtime$.
    \item Otherwise, if $\Gamma\subseteq \IMconj$, then $\numCSPd(\Gamma\cup \GammaPin) \APequiv \numBIS$.
    \item Otherwise there is no FPRAS for $\numCSPd(\Gamma\cup
        \GammaPin)$, unless $\NPtime=\RPtime$.\hfill\qEd 
    \end{itemize}
\end{corollary}

$\Gamma \cup \GammaPin$ is affine (respectively, in \ORconj{} or in
\NANDconj{}) if, and only if $\Gamma$ is, so the case for large-degree
instances ($d\geq 25$) corresponds exactly in complexity to the
unbounded case \cite{DGJ2007:Bool-approx}.  The case for lower degree
bounds is more complex. To put Theorem~\ref{theorem:complexity} in
context, we summarize the known approximability of $\numwHISd$,
parameterized by $d$ and $w$.

The case $d=1$ is clearly in \FPtime{} (Theorem~\ref{thrm:degree-1})
and so is the
case $d=w=2$, which corresponds to counting independent sets in graphs
of maximum degree two. For $d=2$ and width $w\geq 3$, Dyer and
Greenhill have shown that there is an FPRAS for $\numwHISd$
\cite{DG2000:IS-Markov}. For $d=3$, they have shown that there is an
FPRAS if the the width $w$ is at most~3. For larger width, the
approximability of $\numwHISd[3]$ is still not known. With the width
restricted to $w=2$ (normal graphs), Weitz has shown that, for degree
$d\in \{3,4,5\}$, there is a deterministic approximation scheme that
runs in polynomial time (a PTAS) \cite{Wei2006:IS-threshold}. This
extends a result of Luby and Vigoda, who gave an FPRAS for $d\leq 4$
\cite{LV1999:Convergence}.  For $d>5$, approximating $\numwHISd$
becomes considerably harder. More precisely, Dyer, Frieze and Jerrum
have shown that for $d=6$ the Monte Carlo Markov chain technique is
likely to fail, in the sense that ``cautious'' Markov chains are
provably slowly mixing \cite{DFJ2002:IS-sparse}. They also showed
that, for $d=25$, there can be no polynomial-time algorithm for
approximate counting, unless $\NPtime=\RPtime$. These results imply
that for $d\in \{6,\dots,24\}$ and $w\geq 2$ the Monte Carlo Markov
chain technique is likely to fail and for $d\geq 25$ and $w\geq 2$,
there can be no FPRAS unless $\NPtime=\RPtime$.
Table~\ref{tab:aHIS-complexity} summarizes the results.

\begin{table}[t]
\centering\renewcommand{\arraystretch}{1.15}
\begin{tabular}{|c|c|l|}
    \hline
    Degree $d$ & Width $w$ & Approximability of $\numwHISd[d]$ \\
    \hline
    $1$
        & $\geq 2$
        & \FPtime \\
    $2$
        & $2$
        & \FPtime \\
    $2$
        & $\geq 3$
        & FPRAS~\cite{DG2000:IS-Markov} \\
    $3$
        & $2,3$
        & FPRAS~\cite{DG2000:IS-Markov} \\
    $3,4,5$
        & $2$
        & PTAS~\cite{Wei2006:IS-threshold} \\
    $6,\dots,24$
        & $\geq 2$
        & The MCMC method is likely to fail~\cite{DFJ2002:IS-sparse} \\
    $\geq 25$
        & $\geq 2$
        & No FPRAS unless $\NPtime=\RPtime$~\cite{DFJ2002:IS-sparse} \\
    \hline
\end{tabular}
\caption{Approximability of $\numwHISd[d]$ (still open for all other
    values of $d$ and $w$).}
\label{tab:aHIS-complexity}
\end{table}

Returning to bounded-degree \numCSP{}, the case $d=2$ seems to be
rather different to degree bounds three and higher.  This is
also the case for decision CSP --- recall that
degree-$d$ $\CSP(\Gamma\cup\GammaPin)$ has the same complexity as unbounded-degree
$\CSP(\Gamma\cup\GammaPin)$ for all $d\geq 3$
\cite{DF2003:bdeg-gensat}, while degree-2 $\CSP(\Gamma\cup\GammaPin)$
is often easier than the unbounded-degree case
\cite{DF2003:bdeg-gensat,Fed2001:Fanout} but  the complexity of degree-2
$\CSP(\Gamma\cup\GammaPin)$ is still open for some $\Gamma\!$.

Our key techniques for determining the complexity of $\numCSPd(\Gamma
\cup \GammaPin)$ for $d\geq 3$ were the 3-simulation of equality and
Theorem~\ref{thrm:trichotomy}, which says that every Boolean relation
is in \ORconj{}, in \NANDconj{} or 3-simulates equality.  However, it
seems that not all relations that 3-simulate equality also 2-simulate
equality so the corresponding classification of relations does not
appear to hold.  It seems that different techniques will be required
for the degree-2 case.  For example, it is possible that there is no
FPRAS for $\numCSPd[3](\Gamma\cup\GammaPin)$ except when $\Gamma$ is
affine.  However, Bubley and Dyer have shown that there is an FPRAS
for degree-2 \numSAT{}, even though the exact counting problem is
\numP{}-complete \cite{BD1997:Graph-orient}.  This shows that there is
a class $\mathcal{C}$ of constraint languages for which
$\numCSPd[2](\Gamma \cup \GammaPin)$ has an FPRAS for every $\Gamma
\in \mathcal{C}$ but for which no exact polynomial-time algorithm is
known.

We leave the complexity of degree-2 \numCSP{} and of \numBIS{} and the
the various parameterized versions of the counting hypergraph
independent sets problem as open questions.


\bibliographystyle{plain}
\bibliography{dyer}

\end{document}